\documentclass[12pt,a4paper]{article}
\usepackage{amsthm,amsfonts,amsmath,amssymb,array,graphicx,tabularx}
\usepackage[mathscr]{eucal}

\def\supind#1{${}^\mathrm{#1}$}

\def\title#1{\begin{center}{\LARGE\bf #1}\end{center}}

\renewcommand{\author}[2]{\begin{center}{#1}\\ \medskip{\small\it #2}\end{center}}

\newcommand{\I}{\mathrm{i}}

\newcommand{\re}{\mathop{\mathrm{Re}}}
\newcommand{\im}{\mathop{\mathrm{Im}}}
\numberwithin{equation}{section} \theoremstyle{plain}

\newtheorem{prop}{Proposition}[section]

\theoremstyle{definition}

\theoremstyle{remark}

\begin{document}

\title{A mode parabolic equations method with the resonant mode interaction}

\author{Trofimov~M.Yu.\supind{1}, Kozitskiy~S.B.\supind{1}, Zakharenko~A.D.\supind{1,2}}
{\supind{1}Il'ichev Pacific Oceanological Institute, 43 Baltiiskaya St., Vladivostok, 690041, Russia\\
\supind{2} Far Eastern Federal University, 8 Sukhanova str., Vladivostok, 690950, Russia\\
e-mail: {\tt trofimov@poi.dvo.ru, skozi@poi.dvo.ru,
zakharenko@poi.dvo.ru} }

\begin{abstract}
A mode parabolic equation method for resonantly interacted modes was developed.
The flow of acoustic energy is conserved for the derived equations with an accuracy adequate to the used approximation.
The testing calculations were done for ASA wedge benchmark and proved excellent agreement with COUPLE program.
\end{abstract}

\section{Introduction}

Adiabatic acoustic equations appeared as a convenient tool for
solving three-dimensional problems of ocean acoustics since the work
of Collins~\cite{collins} and in refined version since the work of
Trofimov~\cite{tr_na}. Further this method was extended to
interactive modes~\cite{coll2}.
Needless to say, that all this goes back to Burridge and
Weinberg~\cite{Burridge}.
Approach to interacting modes was known outside the parabolic
equation scope~\cite{couple,cracken,abawi}.

Here the method of adiabatic mode parabolic equation is extended to
the case of resonantly interacting modes, more concretely the
considered interaction arises when the wavenumbers of modes are
close to each other. The most intensive interaction of such a type
is observed when the mode of discrete spectrum transforms into the
mode of continuous spectrum, during the propagation, or vice versa.
Such transmutations of modes are common for shallow-water acoustics.
We derive a system of mode parabolic equations which describes this
situation.
It can be easily solved numerically by the Crank-Nicholson implicit
difference scheme in combination with the Gauss-Seidel iteration
method. The application of the corresponding computer code to the
ASA wedge benchmark problem gives the excellent results.

An additional advantage of our method is that the systematic use of
multiscale expansions gives to the applied approach strictness and
completeness.

\section{Basic equations and boundary conditions}

We consider the propagation of time-harmonic sound in the three-dimen\-si\-onal waveguide
$\Omega=\{(x,y,z)| 0 \leq x\leq \infty, -\infty \leq y\leq \infty,  -H\leq z \leq 0 \}$ ($z$-axis is directed upward),
described by the acoustic Helmholtz equation
\begin{equation} \label{Helm}
\left(\gamma P_x\right)_x + \left(\gamma P_y\right)_y
 + \left(\gamma P_z\right)_z + \gamma\kappa^2 P = 0\,,
\end{equation}
where $\gamma = 1/\rho$, $\rho=\rho(x,y,z)$ is the density, $\kappa(x,y,z)$ is the wavenumber.
We assume the appropriate radiation conditions at infinity in $x,y$ plane, the pressure-release
boundary condition at $z=0$
\begin{equation} \label{Dir}
P=0\quad \text{at}\quad z=0\,,
\end{equation}
and rigid boundary condition $\partial u/\partial z=0$ at $z= -H$.
At $x=0$ we impose the Dirichlet boundary condition
\begin{equation} \label{Source}
P=g(z,y)\quad \text{at}\quad x=0\,,
\end{equation}
modelling the sound source located outside $\Omega$.
The parameters of medium may be discontinuous at the nonintersecting smooth interfaces  $z=h_1(x,y),\ldots,h_m(x,y)$,
where the usual interface conditions\sloppy
\begin{equation}\label{InterfCond}
\begin{split}
P_+ = P_-\,,\\
\gamma_+\left(\frac{\partial P}{\partial z}-h_x\frac{\partial P}{\partial x}-
h_y\frac{\partial P}{\partial y}\right)_+ =
\gamma_-\left(\frac{\partial P}{\partial z}-h_x\frac{\partial P}{\partial x}-
h_y\frac{\partial P}{\partial y}\right)_-\,,
\end{split}
\end{equation}
are imposed. Hereafter we use the denotations $f(z_0,x,y)_+=\lim_{z\downarrow z_0}f(z,x,y)$ and
$f(z_0,x,y)_-=\lim_{z\uparrow z_0}f(z,x,y)$.
\par
As will be seen below, it is sufficient to consider the case $m=1$, so we set $m=1$ and denote $h_1$ by $h$.
\par
Assuming that $x$-axis is the preferred direction of propagation, we introduce a small parameter $\epsilon$ (the ratio
of the typical wavelength to the typical size of medium inhomogeneities), the slow
variables $X=\epsilon x$ and $Y=\epsilon^{1/2}y$ (the so called ``parabolic scaling'') and postulate the following
expansions for the parameters $\kappa^2$, $\gamma$ and $h$:
\begin{equation}\label{expans}
\begin{split}
 & \kappa^2 = n_0^2(X,z) + \epsilon\nu(X,Y,z)\,, \\
 & \gamma  =  \gamma_0(X,z) + \epsilon\gamma_1(X,Y,z)\,,\\
 & h = h_0(X) + \epsilon h_1(X,Y)\,.
\end{split}
\end{equation}
To model the attenuation effects we admit $\nu$ to be complex. Namely, we take $\im\nu = 2\eta\beta n_0$, where
$\eta = (40\pi\log_{10}e)^{-1}$ and $\beta$ is the attenuation  in decibels per wavelength. This implies that
$\im\nu\ge 0$.
\par
At first we consider a solution to the Helmholtz equation (\ref{Helm}) in the form of the WKB-ansatz
\begin{equation}\label{ansatz}
P = (u_0(X,Y,z) +\epsilon u_1(X,Y,z) +\ldots)\exp\left(\frac{\mathrm{i}}{\epsilon}\theta(X,Y,z)\right)\,.
\end{equation}
Introducing this anzatz into equation (\ref{Helm}), boundary condition (\ref{Dir}) and interface conditions
(\ref{InterfCond}), all rewritten in the slow variables, we obtain the sequence of the boundary value problems at each
order of $\epsilon$.
\par
From the equations at $O(\epsilon^{-2})$ and $O(\epsilon^{-1})$ we can conclude that $\theta$ is independent of $z$ and $Y$.
Using this information and the Taylor expansion, we can formulate the interface conditions at $h_0$ which are equivalent
to interface conditions (\ref{InterfCond}) up to $O(\epsilon^2)$:
\begin{equation}\label{InterfCondh01}
\left(u_{0}+\epsilon h_1 u_{0z}\right)_+ =(\text{the same terms})_-\,,
\end{equation}

\begin{equation}\label{InterfCondh02}
\begin{split}
\left((\gamma_{0}+\epsilon h_1 \gamma_{0z}+\epsilon \gamma_{1}
\times\left(u_{0z}+\epsilon h_1 u_{0zz}+\epsilon u_{1+}  - \epsilon\mathrm{i}k h_{0X}u_{0}\right)\right)_+ \\
= \left(\mbox{the same terms}\right)_-\,.
\end{split}
\end{equation}

\section{The problem at $O(\epsilon^{0})$}
At $O(\epsilon^{0})$ we obtain
\begin{equation}\label{E0}
(\gamma_0 u_{0z})_z + \gamma_0 n^2_0 - \gamma_0(\theta_X)^2u_0 = 0\,,
\end{equation}
with the interface conditions of the order $\epsilon^{0}$
\begin{equation}\label{InterfaceE0}
u_{0+} = u_0{-}\,,\qquad \left(\gamma_0 \frac{\partial u_0}{\partial z}\right)_+ =
\left(\gamma_0 \frac{\partial u_0}{\partial z}\right)_-\quad\mbox{at}\quad z=h_0\,,
\end{equation}
and boundary conditions $u=0$ at $z=0$ and $\partial u/\partial x$ at $z= -H$.
We seek a solution to problem (\ref{E0}), (\ref{InterfaceE0}) in the form
\begin{equation} \label{anz0}
u_0 = A(X,Y)\phi(X,z)\,.
\end{equation}
From eqs.~(\ref{E0}) and (\ref{InterfaceE0}) we obtain the following spectral problem
for $\phi$ with the spectral parameter $k^2 = (\theta_X)^2$
\begin{equation} \label{Spectral}
\begin{split}
\left(\gamma_0\phi_z\right)_z + \gamma_0 n_0^2 \phi - \gamma_0 k^2 \phi=0\,,\\
\phi(0) = 0\,,\\
\frac{\partial \phi}{\partial z}=0\quad \text{at}\quad z= -H\,,\\
\phi_+ = \phi_-\,,\qquad \left(\gamma_0 \frac{\partial \phi}{\partial z}\right)_+ =
\left(\gamma_0 \frac{\partial \phi}{\partial z}\right)_-\quad\mbox{at}\quad z=h_0\,.
\end{split}
\end{equation}
This spectral problem, considering in the Hilbert space $L_{2,\gamma_0}[-H,0]$ with the scalar product
\begin{equation} \label{L2scalar}
(\phi,\psi) = \int_{-H}^{\,0}\gamma_0 \phi\psi\,dz\,,
\end{equation}
has countably many solutions $(k_j^2,\phi_j)$, $j=1,2,\ldots$ where
the eigenfunction can be chosen as real functions. The eigenvalues
$k_j^2$ are real and have $-\infty$ as a single accumulation
point\cite{nai}.
\par
Let $u_0=A_j(X,Y)\phi_j(X,z)$ where $\phi_j$ is a normalized eigenfunction with the corresponding eigenvalue
$k_j^2>0$ and $A_j$ is an amplitude function to be determined at the next order of $\epsilon$.
The normalizing condition is
\begin{equation} \label{norm}
(\phi,\phi) = \int_{-H}^{\,0}\gamma_0 \phi^2\,dz\,,
\end{equation}

\section{The derivatives of eigenfunctions and wavenumbers with
respect to $X$} Before considering the problem at $O(\epsilon^{1})$
we should consider the problem of calculation the derivatives of
eigenfunctions and wavenumbers with respect to $X$.
\par
Differentiating spectral problem (\ref{Spectral}) with respect to $X$ we obtain the boundary value problem
for $\phi_{jX}$
\begin{equation}\label{phi_X}
\begin{split}
&\left(\gamma_0\phi_{jXz}\right)_z + \gamma_0 n_0^2 \phi_{jX} - \gamma_0 k_j^2 \phi_{jX} = \\
& \qquad\qquad\qquad = - \left(\gamma_{0X}\phi_{jz}\right)_z -
(\gamma_0 n_0^2)_X \phi_{j} + 2k_{jX}k_{j}\gamma_{0} \phi_{j}
+\gamma_{0X} k_j^2 \phi_{j}\,, \\
&\quad  \phi_{jX}(0) = 0\,, \quad  \phi_{jXz}(-H) = 0\,,
\end{split}
\end{equation}
with interface conditions at $z=h_0$
\begin{equation}\label{InterfCond_x}
\begin{split}
\phi_{jX+}-\phi_{jX-} &= -h_{0X}(\phi_{jz+}-\phi_{jz-})\,, \\
\gamma_{0+}\phi_{jXz+} - \gamma_{0-}\phi_{jXz-} &= -\left(\gamma_{0X+}\phi_{jz+}- \gamma_{0X-}\phi_{jz-}\right) - \\
&\qquad - h_{0X}\left(\left((\gamma_{0}\phi_{jz})_z\right)_+ -
\left((\gamma_{0}\phi_{jz})_z\right)_-\right)\,.
\end{split}
\end{equation}
We seek a solution to problem (\ref{phi_X}), (\ref{InterfCond_x}) in the form
\begin{equation}\label{anzX_g}
\phi_{jX} = \sum_{l=0}^\infty C_{jl}\phi_l\,,
\end{equation}
where
\begin{equation}\label{Cjl}
C_{jl} = \int_{-H}^0 \gamma_0 \phi_{jX}\phi_{l}\,dz\,.
\end{equation}
Multiplying (\ref{phi_X}) by $\phi_l$ and then integrating resulting equation from
$-H$ to $0$ by parts twice with the use of interface conditions (\ref{InterfCond_x}), we obtain
\begin{equation}\label{Cjl_2}
\begin{split}
&  \left(k_l^2 - k_j^2\right) C_{jl}  = \int_{-H}^0 \gamma_{0X}
\phi_{jz}\phi_{lz} \,dz  -
\int_{-H}^0 \left(\gamma_0 n_0^2\right)_X \phi_j\phi_l \,dz +  2 k_{jX}k_j\delta_{jl} + \\
&  + k_j^2\int_{-H}^0 \gamma_{0X} \phi_j\phi_l \,dz + \left\{
h_{0X}\gamma_0^2\phi_{jz}\phi_{lz}\left[
\left(\frac{1}{\gamma_{0}}\right)_+
-\left(\frac{1}{\gamma_{0}}\right)_-\right]
-\right. \\
& \qquad\qquad - \left.\left.
h_{0X}\phi_j\phi_l\left[\left(\gamma_0\left(k_j^2-n_0^2\right)\right)_+
 - \left(\gamma_0\left(k_j^2-n_0^2 \right)\right)_-\right]
\vphantom{\frac{1}{\gamma_{0}}}
\right\}\right|_{z=h_0} \,, \\
\end{split}
\end{equation}
where $\delta_{jl}$ is the Kronecker delta. The coefficients $C_{jl}$ can be found from this equation when
$l\ne j$ and at $l=j$ we have the formula for $k_{jX}$.
Differentiating normalizing condition (\ref{norm}) we obtain
\begin{equation}\label{c2s1diff_norm_g}
\begin{split}
&\left(\int_{-H}^0\gamma_0\phi_j^2\,dz\right)_X =
\left(\int_{-H}^{h_0}\gamma_0\phi_j^2\,dz +
\int_{h_0}^{0}\gamma_0\phi_j^2\,dz\right)_X  =\\
& = \int_{-H}^0\gamma_{0X}\phi_j^2\,dz +
2\int_{-H}^0\gamma_{0}\phi_{jX}\phi_j\,dz +
h_{0X}\phi_j^2\left.\left[\gamma_{0-}-\gamma_{0+}\right]\right|_{z=h_0}
= 0\,,
\end{split}
\end{equation}
which gives the equation for $C_{jj}$:
\begin{equation}\label{c2s1C_jj_g}
2C_{jj} =  -\int_{-H}^0\gamma_{0X}\phi_j^2\,dz  +
h_{0X}\phi_j^2\left.\left[\gamma_{0+}-\gamma_{0-}\right]\right|_{z=h_0}\,.
\end{equation}

\section{The problem at $O(\epsilon^{1})$} Let
$\{\theta_j|j=M,\dots,N\}$ be a set of phases. We seek a solution to
the Helmholtz equation (\ref{Helm}) in the form
\begin{equation}\label{ansatzR}
P = \sum_{j=M}^N(u_0^{(j)}(X,Y,z) +\epsilon u_1^{(j)}(X,Y,z) +\ldots)\exp\left(\frac{\mathrm{i}}{\epsilon}\theta_j\right)\,.
\end{equation}
At $O(\epsilon^{1})$ we obtain
\begin{equation}\label{E1}
\begin{split}
&\sum_{j=M}^N\left(\left(\gamma_0u^{(j)}_{1z}\right)_z + \gamma_0 n_0^2 u^{(j)}_{1} - \gamma_0 k_j^2 u^{(j)}_{1}\right)\exp\left(\frac{\mathrm{i}}{\epsilon}\theta_j\right) \\
& = \sum_{j=M}^N\left( -\mathrm{i}\gamma_{0X}k_j u^{(j)}_0 -2\mathrm{i}\gamma_{0}k_j u^{(j)}_{0X} -\mathrm{i}\gamma_{0}k_{jX} u^{(j)}_0 + \gamma_{1}k_j^2 u^{(j)}_0 - \gamma_{0} u^{(j)}_{0YY}\right. \\
& \qquad\quad\,\left. - \frac{\partial}{\partial z}\left(\gamma_{1} u^{(j)}_{0z}\right) - n_0^2\gamma_{1}u^{(j)}_0 -\nu\gamma_{0}u^{(j)}_0\right) \exp\left(\frac{\mathrm{i}}{\epsilon}\theta_j\right) \,,
\end{split}
\end{equation}
with the boundary conditions $u^{(j)}_1=0$ at $z=0$, $\partial u^{(j)}_1/\partial z=0$ at $z=-H$,
and the interface conditions at $z=h_0(X,Y)$ for each $j=M,\dots,N$:
\begin{equation}\label{InterfCondE1}
\begin{split}
 u^{(j)}_{1+}-u^{(j)}_{1-} &= h_1(u^{(j)}_{0z-}-u^{(j)}_{0z+})\,, \\
 \gamma_{0+}u^{(j)}_{1z+}-\gamma_{0-}u^{(j)}_{1z-} &= h_1\left(\left( (\gamma_0 u^{(j)}_{0z})_z\right)_-  -
\left( (\gamma_0 u^{(j)}_{0z})_z\right)_+\right) \\
& + \gamma_{1-}u^{(j)}_{0z-}-\gamma_{1+}u^{(j)}_{0z+}
-\mathrm{i}k_j h_{0X}u^{(j)}_0(\gamma_{0-}-\gamma_{0+})
\end{split}
\end{equation}
We seek a solution to problem (\ref{E1}), (\ref{InterfCondE1}) in the form
\begin{equation}\label{anz1}
u_1^{(j)} = \sum_{l=0}^\infty B_{jl}(X,Y)\phi_l(z,X)\,,
\end{equation}
and introduce coefficients $E_{jl}$ at $j\ne l$ by the equality
\begin{equation}\label{c2s1Ejl_2}
A_j E_{jl} = B_{jl} = \int_{-H}^0 \gamma_0 u_1^{(j)}\phi_l\, dz\,.
\end{equation}
Multiplying (\ref{E1}) by $\phi_l$ and then integrating resulting equation from
$-H$ to $0$ by parts twice with the use of interface conditions (\ref{InterfCondE1}), we obtain
\begin{equation}\label{Bjl_1}
\begin{split}
&\sum_{j=M}^N\left(A_j\cdot \vphantom{\left.
\left(\frac{\gamma_1}{\gamma_0}\right)\right|_{z=h_0}}\left\{h_1\phi_l\left[\left((\gamma_0\phi_{jz})_z\right)_+ -
\left((\gamma_0\phi_{jz})_z\right)_- \right] \vphantom{
\left(\frac{\gamma_1}{\gamma_0}\right)}
\right.\right.\\
&\qquad +
\gamma_0\phi_{jz}\phi_l\left[\left(\frac{\gamma_1}{\gamma_0}\right)_+
- \left(\frac{\gamma_1}{\gamma_0}\right)_-\right]
-\mathrm{i}k_jh_{0X}\phi_j\phi_l\left[\gamma_{0+}-\gamma_{0-}\right] \\
&\qquad\left.\left.\left.-
h_1\gamma_0^2\phi_{jz}\phi_{lz}\left[\left(\frac{1}{\gamma_0}\right)_+
-
\left(\frac{1}{\gamma_0}\right)_-\right]\right\}\right|_{z=h_0} + A_j(k_l^2-k_j^2)E_{jl}\right)\exp\left(\frac{\mathrm{i}}{\epsilon}\theta_j\right) \\
& = \sum_{j=M}^N\left(-\mathrm{i}k_j A_j \int_{-H}^0 \gamma_{0X}\phi_j\phi_l\,dz
-2\mathrm{i}k_j A_j \int_{-H}^0 \gamma_{0}\phi_{jX}\phi_l\,dz \right. \\
& \quad -2\mathrm{i}k_j A_{j,X} \int_{-H}^0 \gamma_{0}\phi_{j}\phi_l\,dz
- \mathrm{i}k_{j,X} A_j \int_{-H}^0 \gamma_{0}\phi_{j}\phi_l\,dz \\
& \quad + k_j^2 A_j \int_{-H}^0 \gamma_{1}\phi_{j}\phi_l\,dz
-  A_{j,YY} \int_{-H}^0 \gamma_{0}\phi_{j}\phi_l\,dz \\
& \quad - A_j \int_{-H}^0 \left(\gamma_{1}\phi_{jz}\right)_z\phi_l\,dz
- A_j \int_{-H}^0 \gamma_{1}n_0^2\phi_{j}\phi_l\,dz \\
& \quad \left.- A_j \int_{-H}^0 \nu\gamma_{0}\phi_{j}\phi_l\,dz\right)\exp\left(\frac{\mathrm{i}}{\epsilon}\theta_j\right) \,.
\end{split}
\end{equation}
The terms $A_j(k_l^2-k_j^2)E_{jl}$ in these expressions can be omitted
because of the resonant condition $|k_l-k_j|<\epsilon$.

As
\begin{equation*}
\begin{split}
& -\mathrm{i}k_j A_j \int_{-H}^0 \gamma_{0X}\phi_j\phi_l\,dz
-2\mathrm{i}k_j A_j \int_{-H}^0 \gamma_{0}\phi_{jX}\phi_l\,dz  \\
& = \mathrm{i}k_j A_j\left(C_{lj} - C_{jl}\right)
-\mathrm{i}k_jA_jh_{0X}\phi_j\phi_l\left.\left[\gamma_{0+}-\gamma_{0-}\right]\right|_{z=h_0}\,,
\end{split}
\end{equation*}
and
\begin{equation*}
\begin{split}
- A_j \int_{-H}^0 \left(\gamma_{1}\phi_{jz}\right)_z\phi_l\,dz =
 &A_j \int_{-H}^0 \gamma_{1}\phi_{jz}\phi_{lz}\,dz \\
 + &A_j\gamma_0\phi_{jz}\phi_l\left[\left(\frac{\gamma_1}{\gamma_0}\right)_+ -
 \left(\frac{\gamma_1}{\gamma_0}\right)_-\right]\,,
\end{split}
\end{equation*}
we obtain, after some algebra,
\begin{equation}\label{Bjl_2}
\begin{split}
&\sum_{j=M}^N\left(A_j\vphantom{\left.\left.\left.\left(\frac{1}{\gamma_0}\right)_-\right]\right\}\right|_{z=h_0}  }  \cdot\left\{h_1\phi_l\left[\left((\gamma_0\phi_{jz})_z\right)_+ -
\left((\gamma_0\phi_{jz})_z\right)_- \right] \vphantom{
\left(\frac{\gamma_1}{\gamma_0}\right)_-} -
\right.\right.\\
&\qquad\left.\left.\left.-
h_1\gamma_0^2\phi_{jz}\phi_{lz}\left[\left(\frac{1}{\gamma_0}\right)_+
-
\left(\frac{1}{\gamma_0}\right)_-\right]\right\}\right|_{z=h_0} \right)\exp\left(\frac{\mathrm{i}}{\epsilon}\theta_j\right) \\
& = \sum_{j=M}^N\left( \vphantom{\int_{-H}^0 \nu\gamma_{0}\phi_{j}\phi_l\,dz}\mathrm{i}k_j A_j\left(C_{lj} - C_{jl}\right)
 -2\mathrm{i}k_j A_{j,X} \delta_{jl} - \mathrm{i}k_{jX} A_j\delta_{jl} -  A_{j,YY} \delta_{jl} \right.\\
& \quad + k_j^2 A_j \int_{-H}^0 \gamma_{1}\phi_{j}\phi_l\,dz
 \quad + A_j \int_{-H}^0 \gamma_{1}\phi_{jz}\phi_{lz}\,dz
- A_j \int_{-H}^0 \gamma_{1}n_0^2\phi_{j}\phi_l\,dz  \\
& \left.\quad - A_j \int_{-H}^0 \nu\gamma_{0}\phi_{j}\phi_l\,dz \right)\exp\left(\frac{\mathrm{i}}{\epsilon}\theta_j\right)\,.
\end{split}
\end{equation}

\begin{prop}\label{pr_MPE}
The solvability condition for the problem at $O(\epsilon^1)$  is a system of parabolic wave equtions for $l=M,\ldots,N$
\begin{equation}\label{MPE}
2\mathrm{i}k_l A_{l,X} + \mathrm{i}k_{l,X} A + A_{l,YY} + \alpha_{ll} A_l + \sum_{j=M,j\ne l}^N\alpha_{lj}A_j\exp(\theta_{lj})= 0\,,
\end{equation}
where $\alpha_{lj}$ is given by the following formula
\begin{equation}\label{alpha}
\begin{split}
\alpha_{lj} & = \int_{-\infty}^0 \gamma_0\nu \phi_j\phi_l\,dz +
\int_{-\infty}^0 \gamma_1\left(n_0^2-k_j^2\right) \phi_j\phi_l\,dz -
\int_{-\infty}^0 \gamma_1\phi_{jz}\phi_{lz}\,dz  \\
& -\mathrm{i}k_j\left(C_{lj} - C_{jl}\right)\\
& + \left\{h_1\phi_l\left[\left((\gamma_0\phi_{jz})_z\right)_+ - \left((\gamma_0\phi_{jz})_z\right)_- \right]
\vphantom{ \left(\frac{\gamma_1}{\gamma_0}\right)_-}
\right.\\
& \qquad\qquad\qquad\left.\left.- h_1\gamma_0^2\phi_{jz}\phi_{lz}\left[\left(\frac{1}{\gamma_0}\right)_+ -
\left(\frac{1}{\gamma_0}\right)_-\right]\right\}\right|_{z=h_0}
\,,
\end{split}
\end{equation}
and
\begin{equation}\label{detun}
\theta_{lj}=\mathrm{i}(\theta_{l}-\theta_{j})
\end{equation}
\end{prop}

Using spectral problem~(\ref{Spectral}) the interface terms in (\ref{alpha}) can be rewritten also as
\begin{equation*}
\begin{split}
&  \left\{h_1\phi_j\phi_l\left[k_j^2\left(\gamma_{0+}-\gamma_{0-}\right)
- \left(n_0^2\gamma_0\right)_+ + \left(n_0^2\gamma_0\right)_-  \right]
\vphantom{ \left(\frac{\gamma_1}{\gamma_0}\right)_-}
\right.\\
&\qquad\qquad\qquad \left.\left.- h_1\gamma_0^2\phi_{jz}\phi_{lz}\left[\left(\frac{1}{\gamma_0}\right)_+ -
\left(\frac{1}{\gamma_0}\right)_-\right]\right\}\right|_{z=h_0}
\,.
\end{split}
\end{equation*}
We shall refer to the quantities and variables $X$, $Y$, $\nu$,
$\theta_j$, and $A_j$ as the {\em asymptotic} ones. Considerations
of initial-boundary value problems in a (partially) bounded domain
require the use of {\em physical} quantities and variables, which
will be $x$, $y$, $\bar\nu=\epsilon\nu$, $\bar A_j(x,y)=A_j(\epsilon
x,\sqrt{\epsilon}y)=A_j(X,Y)$, and
$\displaystyle{\bar\theta_j=\int^X\frac{1}{\epsilon}\theta_{j,X}\,dX=\int^x
k_j\,dx}$.
 It can be easily verified that equations~(\ref{MPE}) in physical variables has the same form
\begin{equation}\label{MPEph}
2\mathrm{i}k_l \bar A_{l,x} + \mathrm{i}k_{l,x} \bar A_j + \bar A_{l,yy} +\bar \alpha_{ll} \bar A_l + \sum_{j=M,j\ne l}^N\bar \alpha_{lj} \bar A_j\exp(\bar\theta_{lj}) = 0\,,
\end{equation}
where $\bar \alpha_{lj}$ are expressed by the same formulas as $\alpha_{lj}$ with $\nu$ replaced by $\bar\nu$,
$\bar\theta_{lj}=(\bar\theta_{l}-\bar\theta_{j})$.

\section{Initial-boundary value problems for mode parabolic
equation} For eq.~\ref{MPEph} we shall consider the initial-boundary
value problem in domain of the form $\{(x,y)|0\le x <\infty\,,Y_1\le
y \le Y_2\}$ with the initial condition
\begin{equation}\label{in_cond}
\bar A_j(0,y)=g_j(y)=(g,\phi_j)=\int_{-H}^{\,0} \gamma_0 g(z,y)\phi_j(z)\,dz\,,
\end{equation}
interface conditions~(\ref{intVPar}) and transparent boundary conditions at $y=Y_1$ and $y=Y_2$.

\subsection{Vertical interfaces and boundaries and corresponding interface and boundary conditions}
We consider vertical interfaces along smooth curves of the form $\{(x,y)|y={\cal{I}}(x)\}$.
Such an interface is formed mostly by the jump of topography $h_1$ at $y={\cal{I}}(x)$.
The usual interface conditions for eq.~(\ref{Helm}) are
\begin{equation} \label{intV}
\begin{split}
P|_{y={\cal{I}}(x)+0}=P|_{y={\cal{I}}(x)-0}\,, \quad
\left.\gamma\frac{\partial P}{\partial n}\right|_{y={\cal{I}}(x)+0}
= \left.\gamma\frac{\partial P}{\partial n}\right|_{y={\cal{I}}(x)+0}\,,
\end{split}
\end{equation}
where $\partial/\partial n$ denotes the normal derivative.
Assuming that $P=\exp(\bar\theta_j)\bar A_j\phi_j$,
we have
\begin{equation} \label{intVP}
\begin{split}
\gamma\frac{\partial P}{\partial n} = \gamma P_y-\gamma{\cal{I}}_x P_x =
\gamma_0\exp(\bar\theta)\left[A_{j,y}- \I{\cal{I}}_x k_j A_j \right]\phi_j+ O(\epsilon)\,.
\end{split}
\end{equation}
Multiplying eq.~(\ref{intVP}) by $\phi_j$ and integrating with respect to $z$, we have
\begin{equation*}
\begin{split}
\int_{-H}^{\,0}\gamma\frac{\partial P}{\partial n}\,dz
=\exp(\bar\theta)\left[A_{j,y}- \I{\cal{I}}_x k_j A_j \right] + O(\epsilon) \,,
\end{split}
\end{equation*}
The interface conditions at $y={\cal{I}}(x)$ modulo $\epsilon$ now become
\begin{equation} \label{intVPar}
\begin{split}
A_{j\,L}=A_{j\,R}\,,\\
\left(A_{j,y}- \I{\cal{I}}_x k_j A_j\right)_L=\left(A_{j,y}- \I{\cal{I}}_x k_j A_j\right)_R\,,
\end{split}
\end{equation}
where we use the denotations $f(x_0,y_0)_R=\lim_{y\downarrow y_0}f(x_0,y)$ and
$f(x_0,y_0)_L=\lim_{y\uparrow y_0}f(x_0,y)$, $(x_0,y_0)$ is the interface point.
As $k_j$ is assumed to be continuous through the interface, then finally the interface conditions take the form
\begin{equation} \label{intVParF}
\begin{split}
A_{j\,L}=A_{j\,R}\,,\\
\left(A_{j,y}\right)_L=\left(A_{j,y}\right)_R\,.
\end{split}
\end{equation}
The analogous considerations give the following boundary conditions at the boundary $y={\cal{B}}(x)$:
\[
A_j|_{y={\cal{B}}(x)}=0
\]
at the soft boundary and
\begin{equation} \label{boundPar}
\left.\left(A_{j,y}- \I{\cal{B}}_x k_j A_j\right)\right|_{y={\cal{B}}(x)}=0
\end{equation}
at the rigid boundary.

\section{Energy flux conservation for parabolic
equations~(\ref{MPEph})} The time averaged acoustic energy flux
through the plane $x=x_0$ is defined as
\[
J(x_0)=\frac{1}{2\omega}\im\int_{-\infty}^\infty I(x_0,y)\,dy\,,
\]
where
\[
I(x_0,y)=\int_{-H}^{\,0} \gamma P_x(x_0,y,z)P^*(x_0,y,z)\,dz\,,
\]
$P^*$ is the complex conjugate of $P$.
\par
As is well known, if $P$ is a solution of the Helmholtz equation~(\ref{Helm}) then the corresponding energy flux
is conserved, that is
\[
\frac{dJ}{dx}=0\,.
\]
\begin{prop}\label{pr_flux}
Assume that $\im\bar\nu=0$. Let $\{\bar A_j|j=M,\ldots N\}$ be a solution to equations~(\ref{MPEph}) with interface
conditions~(\ref{intVParF}) and boundary condition~(\ref{boundPar}) .
Then for
$\displaystyle{ P = \sum_{j=M}^N\bar A_j\exp(\I\bar\theta)\phi_j= \sum_{j=M}^N A_j\exp\left(\frac{\I}{\epsilon}\theta\right)}\phi_j$
\[
\frac{dJ}{dx}=O(\epsilon^2)
\]
\end{prop}
\begin{proof}
First calculate the derivative of the flux with respect to $x$
for the anzats used:
\begin{equation}\label{eFlux}
\begin{split}
& 2\omega\frac{dJ}{dx} 
= \frac{d}{dx}\int_{-\infty}^\infty\left[\sum_{l=M}^N
k_l|A_l|^2\right.\\
& +\epsilon\sum_{l=M}^N\sum_{j=M}^N\im\left(C_{lj}A_lA_j^*\exp\left(\frac{\I}{\epsilon}(\theta_l-\theta_j)\right)\right)
 \left.+\epsilon\sum_{l=M}^N \im(A_{l,X}A_l^*)\right]\,dy\\
& =\int_{-\infty}^\infty
\left[\epsilon\sum_{l=M}^N\sum_{j=M}^N(k_l-k_j)C_{lj}\re\left(A_jA_l^*\exp\left(\frac{\I}{\epsilon}(\theta_j-\theta_l)\right)\right)\right.\\
& \left.\qquad +\epsilon\sum_{l=M}^N (k_l|A_l|^2)_X\right] \,dy
+O(\epsilon^2)\,.
\end{split}
\end{equation}
Consider now the sum on $l$ of the equations (\ref{MPE}) multiplied by $A_l^*$ subtracted
conjugate equation multiplied by $A_l$ and integrate the result on $y$ from minus infinity to infinity:
\begin{equation*}
\begin{split}
& \sum_{l=M}^N\int_{-\infty}^\infty\left[\vphantom{\sum_{l=M}^N}\left(2\mathrm{i}k_l A_{l,X} + \mathrm{i}k_{l,X} A_l +
A_{l,YY}+\sum_{j=M}^N\alpha_{lj}A_j\exp(\theta_{lj})\vphantom{\frac{1}{R}}\right)A_l^*\right. - \\
& \left((-2\mathrm{i}k_l A^*_{l,R} - \mathrm{i}k_{l,R}A^*_l 
+A^*_{l,YY}
+\sum_{j=M}^N\left.\alpha^*_{lj}A^*_j\exp(\theta^*_{lj})\vphantom{\frac{1}{R}}\right)A_l\right]\,dy= 0\,.
\end{split}
\end{equation*}
Further the terms with the second derivative with respect to $Y$ vanish due to boundary conditions.
After some transformation we have:
\begin{equation*}
\begin{split}
&\sum_{l=M}^N\sum_{j=M}^N\int_{-\infty}^\infty(\alpha_{lj}A_j\exp(\theta_{lj})A_l^*
 - \alpha_{lj}^*A_j^*\exp(\theta_{lj}^*)A_l)\,dy\\
&\qquad\qquad+\sum_{l=M}^N 2\mathrm{i}\int_{-\infty}^\infty(k_l|A_{l}|^2 )_X\,dy=0\,,
\end{split}
\end{equation*}
then substitute for $\alpha_{lj}$ its expression (\ref{alpha})
\begin{equation*}
\begin{split}
&\sum_{l=M}^N\sum_{j=M}^N\int_{-\infty}^\infty
\left(-\mathrm{i}k_j(C_{lj}-C_{jl})A_j\exp(\frac{\I}{\epsilon}(\theta_j-\theta_l))A_l^* + \right. \\
&\left. -\mathrm{i}k_j(C_{lj}-C_{jl}) A_j^* \exp(\frac{\I}{\epsilon}(\theta_l-\theta_j)) A_l \right)\,dy
+\sum_{l=M}^N 2\mathrm{i}\int_{-\infty}^\infty(k_l|A_{l}|^2 )_X\,dy =0\,,
\end{split}
\end{equation*}
and collect terms
\begin{equation*}
\begin{split}
&\sum_{l=M}^N\sum_{j=M}^N\int_{-\infty}^\infty\left(
-\mathrm{i}k_j(C_{lj}-C_{jl})2\re(A_j\exp(\frac{\I}{\epsilon}(\theta_j-\theta_l))A_l^*)\right)\,dy\\
&\qquad\qquad+\sum_{l=M}^N 2\mathrm{i}\int_{-\infty}^\infty(k_l |A_{l}|^2 )_X\,dy = 0\,,
\end{split}
\end{equation*}
write double sums separately for terms with $C_{lj}$ and $C_{jl}$
\begin{equation*}
\begin{split}
&\sum_{l=M}^N\sum_{j=M}^N\int_{-\infty}^\infty\left(
-\mathrm{i}k_jC_{lj}2\re(A_j\exp(\frac{\I}{\epsilon}(\theta_j-\theta_l))A_l^*)\right)\,dy\\
&\qquad+\sum_{l=M}^N\sum_{j=M}^N\int_{-\infty}^\infty\left(
\mathrm{i}k_jC_{jl}2\re(A_j\exp(\frac{\I}{\epsilon}(\theta_j-\theta_l))A_l^*)\right)\,dy\\
&\qquad\qquad\qquad+\sum_{l=M}^N 2\mathrm{i}\int_{-\infty}^\infty(k_l|A_{l}|^2 )_X\,dy = 0\,,
\end{split}
\end{equation*}
exchange indexes $l$ and $j$ in the second double sum and finally get
\begin{equation*}
\begin{split}
&\sum_{l=M}^N\sum_{j=M}^N \int_{-\infty}^\infty\left(
\mathrm{i}(k_l-k_j)C_{lj}2\re(A_j\exp(\frac{\I}{\epsilon}(\theta_j-\theta_l))A_l^*)\right)\,dy\\
&\qquad\qquad\qquad+\sum_{l=M}^N 2\mathrm{i}\int_{-\infty}^\infty(k_l |A_{l}|^2 )_X\,dy = 0\,,
\end{split}
\end{equation*}
The last equation coincides modulo $2\mathrm{i}$ with the $O(\epsilon)$-part of (\ref{eFlux}).
\end{proof}

\begin{figure}
\begin{center}
\includegraphics[width=\textwidth]{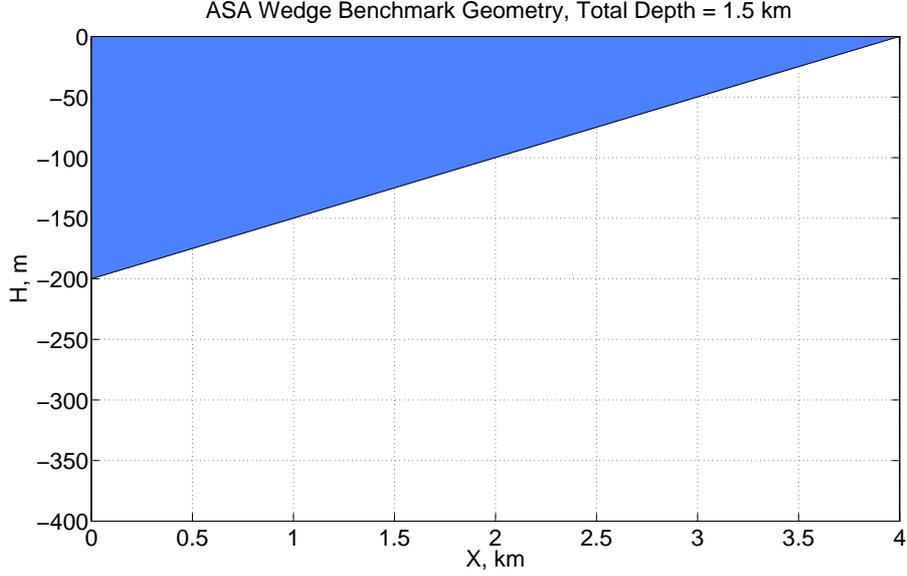}
\caption{ASA wedge benchmark geometry. Harmonic point source with
frequency 25 Hz is located at depth 100 m. Sound speed in the water
layer is 1.5 km/s, in the bottom is 1.7 km/s. Density of the water
is 1000 kg/m$^3$, of the bottom is 1500 kg/m$^3$. Attenuation in the
water is absent, in the bottom is 0.5 dB/$\lambda$ till the depth of
1 km, then attenuation linearly increases up to 2.5 dB/$\lambda$ at
depth 1.5 km.}
\end{center}
\end{figure}

\begin{figure}
\begin{center}
\includegraphics[width=\textwidth]{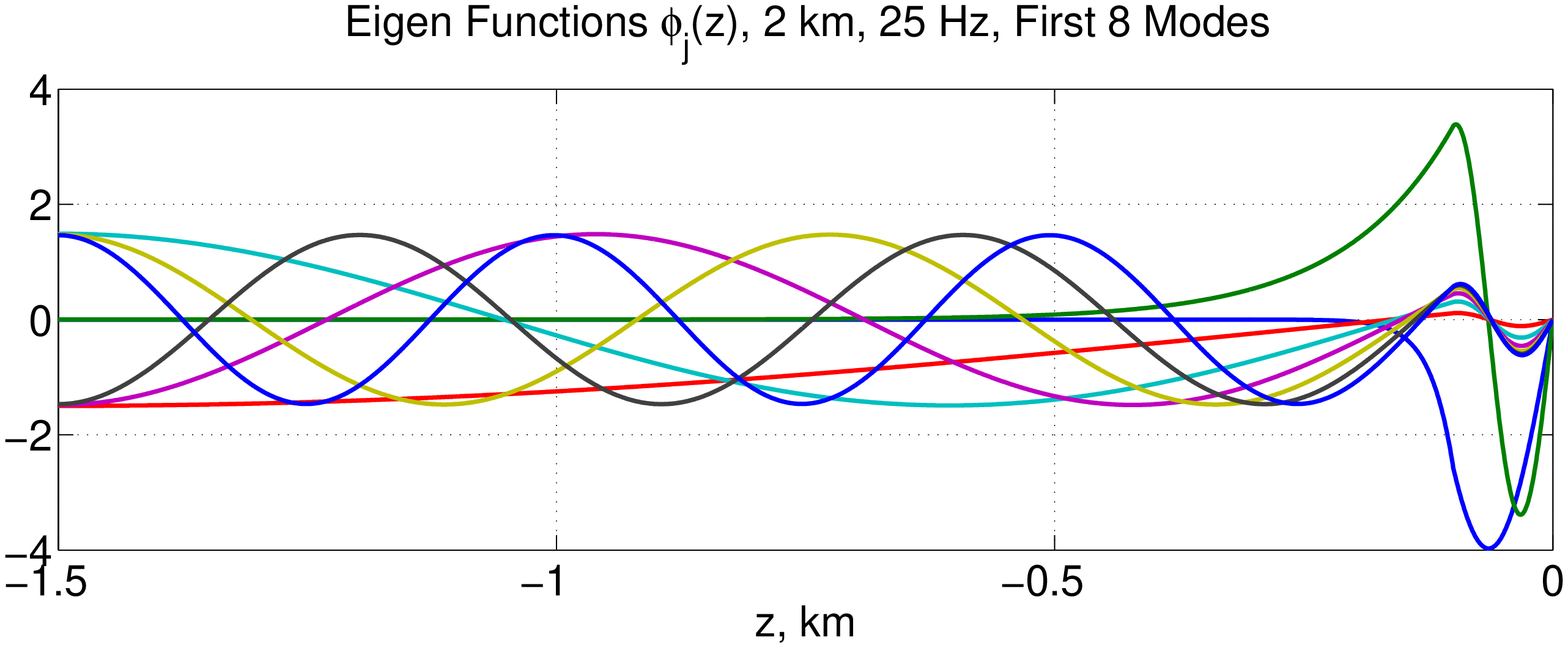}
\includegraphics[width=\textwidth]{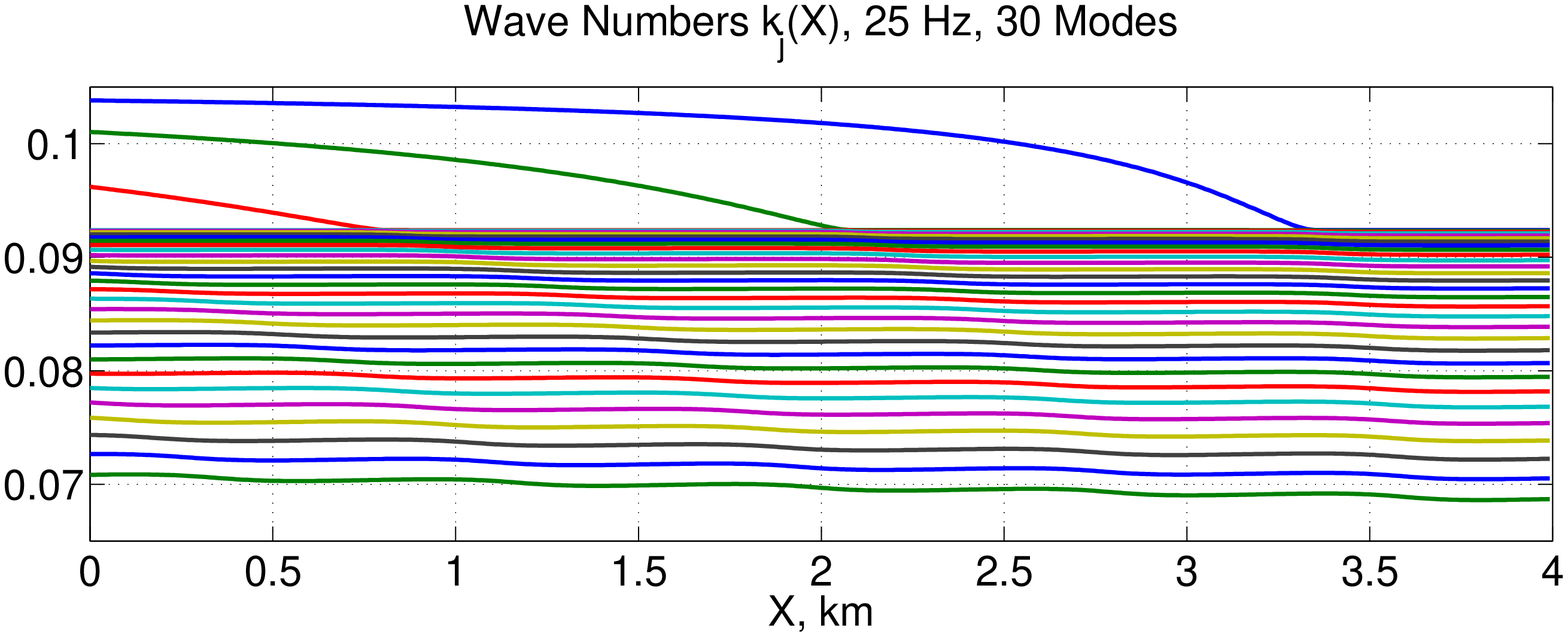}
\caption{Mode waveforms  at distance X=2 km and wavenumbers $k_j(X)$
for the ASA wedge.}
\end{center}
\end{figure}

\begin{figure}
\begin{center}
\includegraphics[width=\textwidth]{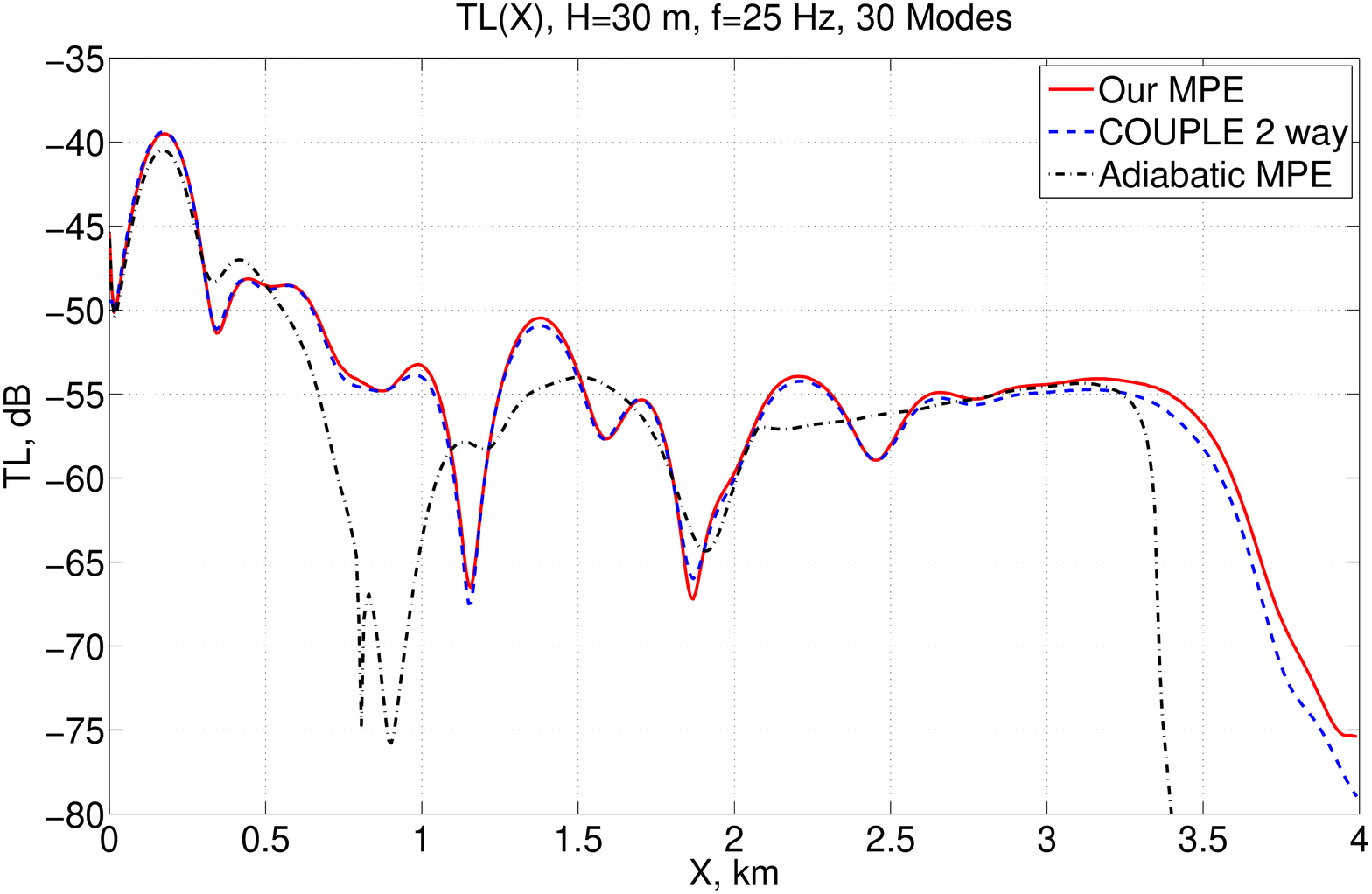}
\caption{Transmission loss for the ASA wedge, receiver depth=30 m.
Comparison with COUPLE program (Meansqu. error=0.9 dB) and Adiabatic
MPE.}
\end{center}
\end{figure}

\begin{figure}
\begin{center}
\includegraphics[width=\textwidth]{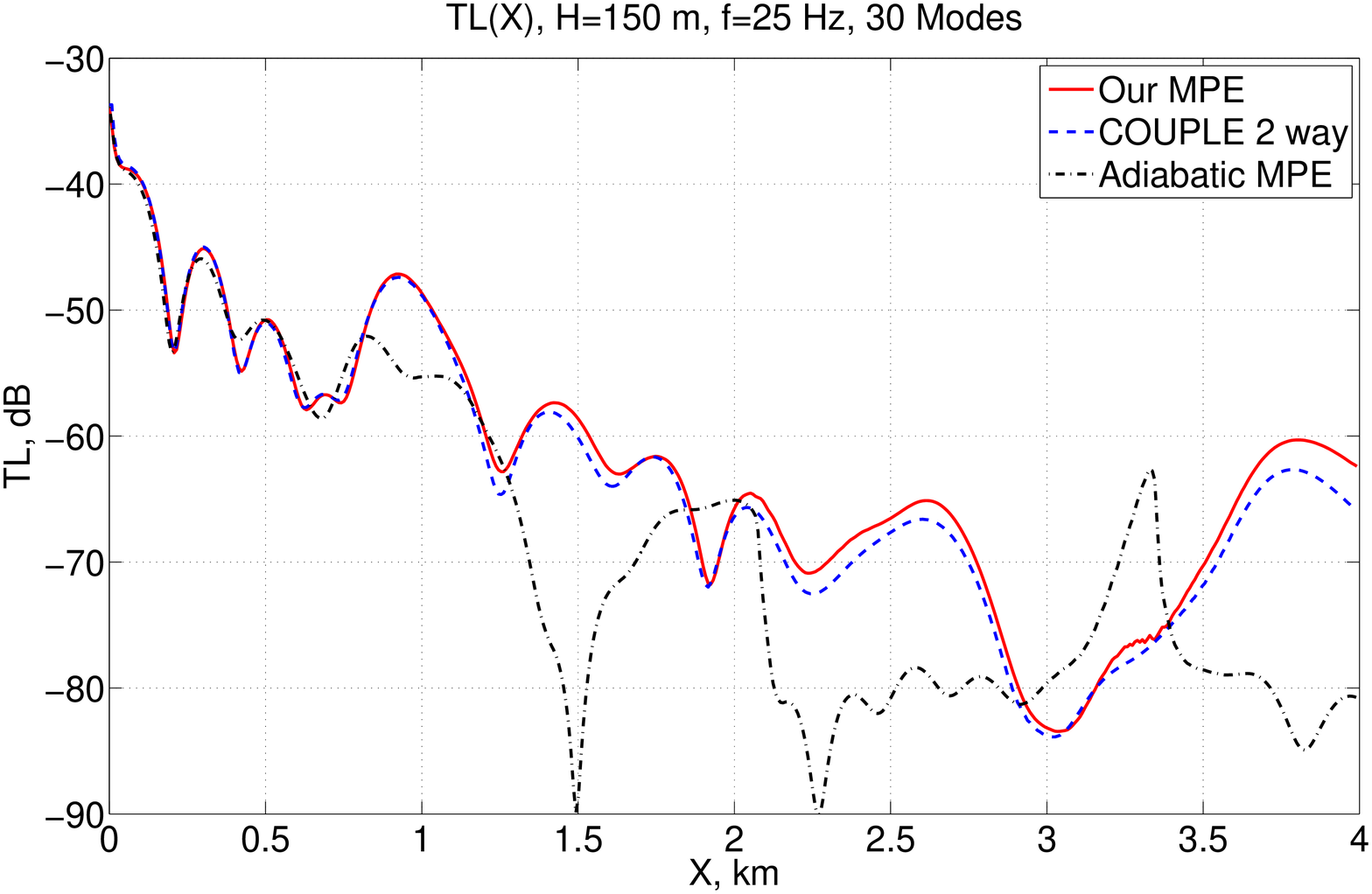}
\caption{Transmission loss for the ASA wedge, receiver depth=150 m.
Comparison with COUPLE program (Meansqu. error=1.3 dB) and Adiabatic
MPE.}
\end{center}
\end{figure}

\begin{figure}
\begin{center}
\includegraphics[width=\textwidth]{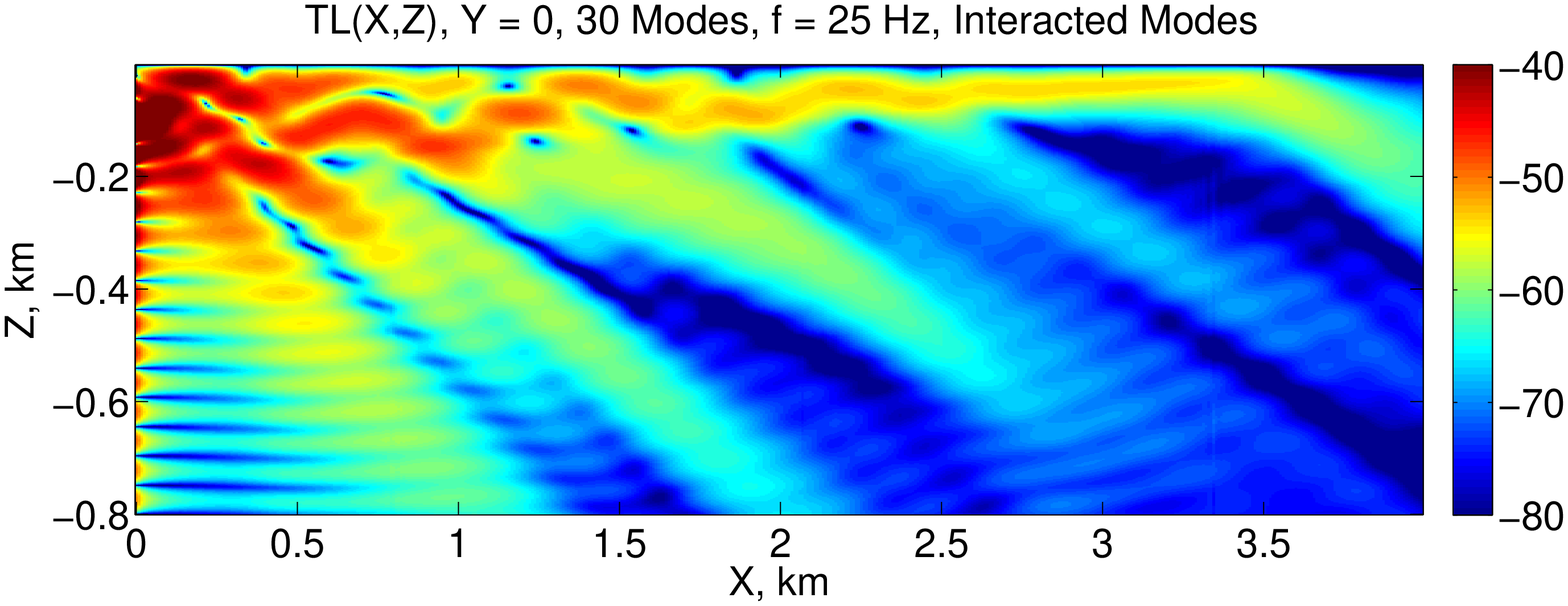}
\includegraphics[width=\textwidth]{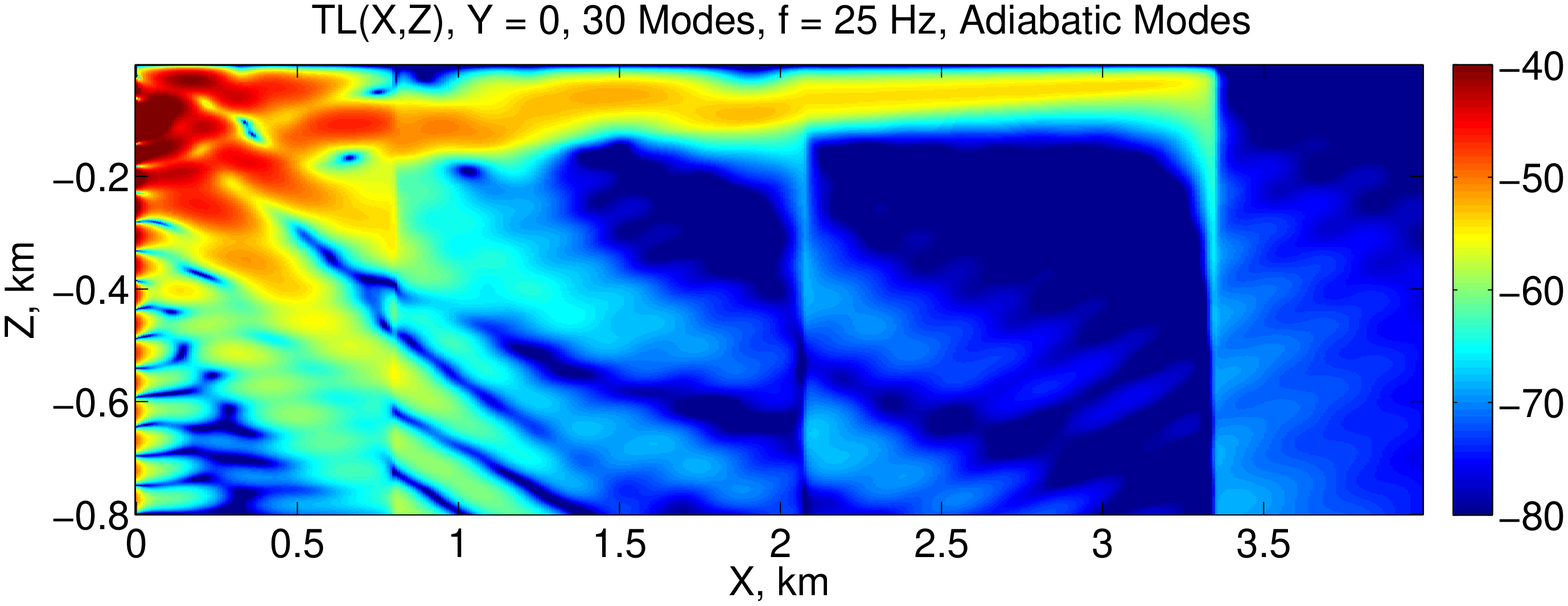}
\caption{Transmission loss, XZ plain. Top: interacted modes, bottom:
adiabatic modes}
\end{center}
\end{figure}

\begin{figure}
\begin{center}
\includegraphics[width=\textwidth]{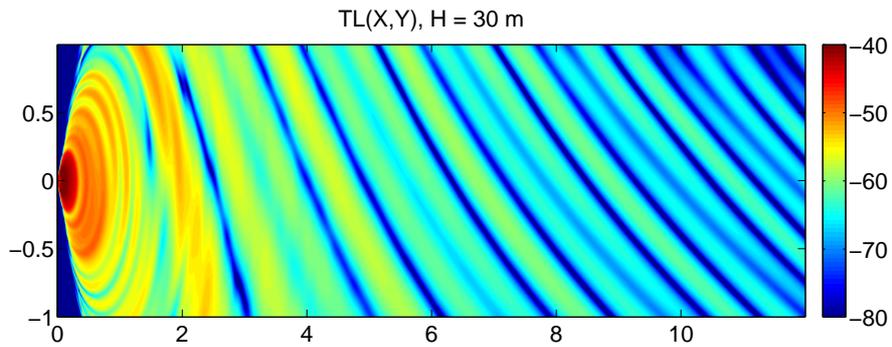}
\caption{Transmission loss for the ASA wedge, plain XY, receiver
depth=30 m. Across slope propagation. Adiabatic modes. Total depth
is 600 m, wedge bottom depth is 200 m, Attenuation in the bottom is
0.5 dB/$\lambda$ till the depth of 500~m, then attenuation linearly
increases up to 5.5 dB/$\lambda$ at depth 600~m. The other
parameters are the same as in classical ASA wedge benchmark.}
\end{center}
\end{figure}

\begin{figure}
\begin{center}
\includegraphics[width=\textwidth]{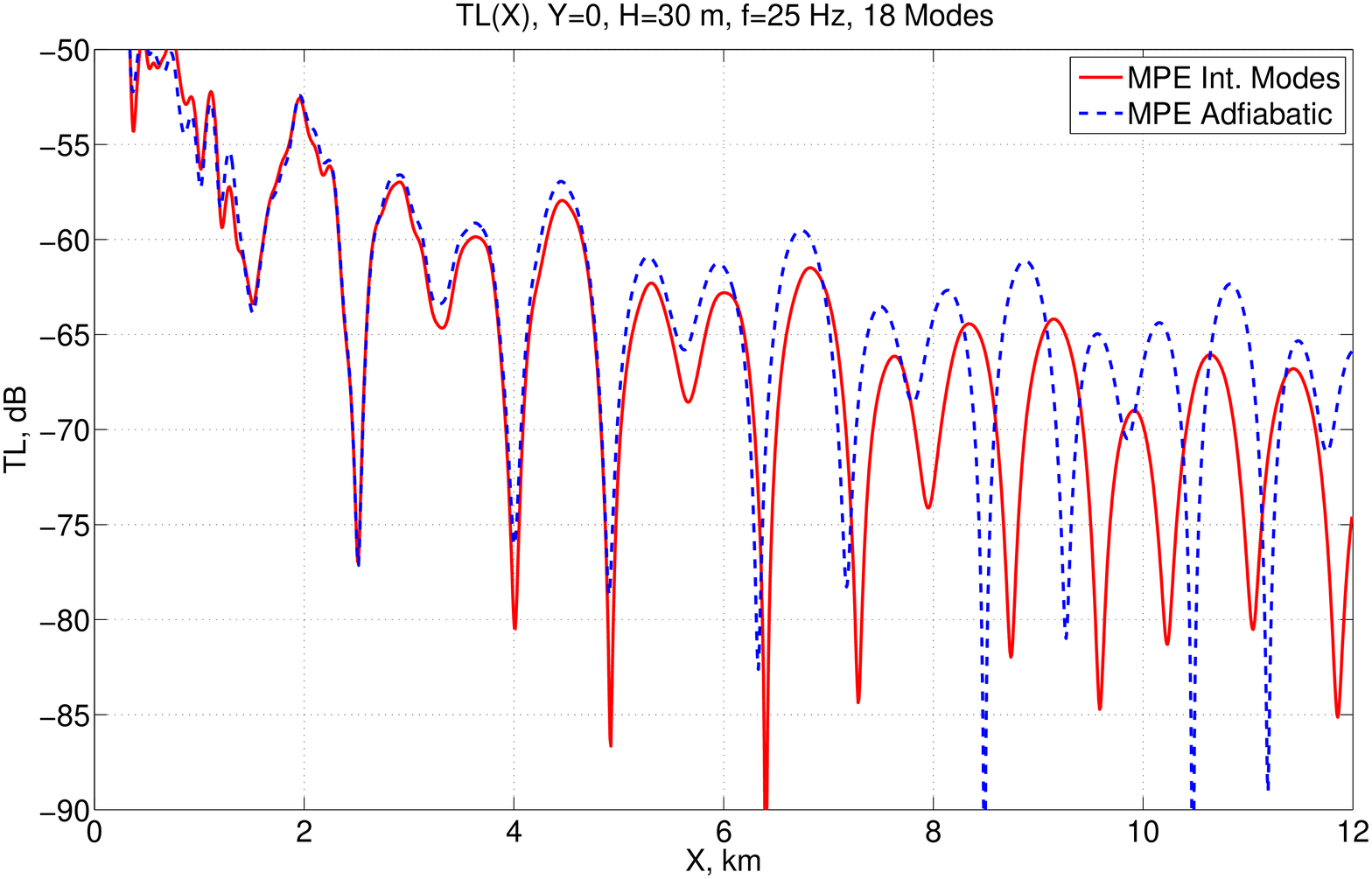}
\caption{Transmission loss for the ASA wedge, receiver depth=30 m.
Across slope propagation. Interacted modes vs. adiabatic modes.}
\end{center}
\end{figure}

\section{Conclusion}

In this article a mode parabolic equation method for resonantly interacted modes was developed.
The proposed method is an essential extension of the early proposed method of adiabatic mode parabolic
equation because it can serve all possible problems of shallow water acoustics.
The flow of acoustic energy is conserved for the derived equations with an accuracy adequate to the used approximation.
The proposed method was tested.
The testing calculations were done for ASA wedge benchmark and proved excellent agreement with COUPLE program.

\end{document}